\documentclass[sn-mathphys]{sn-jnl}
\usepackage{enumerate}
\usepackage{float}
\usepackage{graphicx}
\usepackage{subfigure}
\usepackage{algorithm}
\jyear{2022}

\theoremstyle{thmstyleone}
\newtheorem{theorem}{Theorem}

\theoremstyle{thmstyletwo}

\theoremstyle{thmstylethree}

\begin{document}

\title[W.J.Liu et al.]{Public Verifiable Measurement-Only Blind Quantum Computation Based On Entanglement Witnesses}

\author*[1,2]{\fnm{Wen-Jie} \sur{Liu}}\email{wenjiel@163.com}

\author[1]{\fnm{Zi-Xian} \sur{Li}}\email{zixianli157@163.com}
\equalcont{These authors contributed equally to this work.}

\author[1]{\fnm{Wen-Bo} \sur{Li}}\email{lwb0408@163.com}
\equalcont{These authors contributed equally to this work.}

\author[1]{\fnm{Qi} \sur{Yang}}\email{yangqi7275@163.com}
\equalcont{These authors contributed equally to this work.}

\affil*[1]{\orgdiv{School of Software}, \orgname{Nanjing University of Information Science and Technology}, \orgaddress{\street{No. 219 Ningliu Road}, \city{Nanjing}, \postcode{210044}, \state{Jiangsu}, \country{China}}}

\affil[2]{\orgdiv{Engineering Research Center of Digital Forensics}, \orgname{Ministry of Education}, \orgaddress{\street{No. 219 Ningliu Road}, \city{Nanjing}, \postcode{210044}, \state{Jiangsu}, \country{China}}}

\abstract{Recently Sato et al. proposed an public verifiable blind quantum computation (BQC) protocol by inserting a third-party arbiter. However, it isn’t true public verifiable in a sense, because the arbiter is determined in advance and participates in the whole process. In this paper, a public verifiable protocol for measurement-only BQC is proposed. The fidelity between arbitrary states and the graph states of $2$-colorable graphs is estimated by measuring the entanglement witnesses of the graph states, so as to verify the correctness of the prepared graph states. Compared with the previous protocol, our protocol is public verifiable in the true sense by allowing other random clients to execute the public verification. It also has greater advantages in the efficiency, where the number of local measurements is $O(n^3\log{n})$ and graph states' copies is $O(n^2\log{n})$.}
 

\maketitle

\section{Introduction}\label{sec1}

%找句式重复、用词重复、语法不对的找出来，十个左右

Blind quantum computation (BQC) allows any client (known as Alice) with weak quantum ability to delegate her computing tasks to a quantum server (known as Bob) without leaking her privacy. BQC is divided into two categories: circuit-based BQC (CBQC)\cite{RN1, RN2,RN3,RN4,RN5} and measurement-based BQC (MBQC)\cite{RN6, RN7, RN8, RN9, RN10, RN11, RN12, RN13, RN14,RN15, RN17,RN18}. CBQC realizes blindness through quantum circuits, where the client needs to have the ability to operate some quantum gates. In MBQC, the client only needs to prepare and measure quantum states. Recently, Morimae and Fuji\cite{RN7} proposed a new type of BQC, called measurement-only BQC (MOBQC), where the server prepare the resource state, while the client just should perform single-qubit measurements.

As more and more BQC protocols are proposed, the \textit{verifiability} of BQC has attracted much attention. In a verifiable BQC protocol, each party can verify whether other party is honest. Although Broadbent et al.\cite{RN6} have explored the possibility of verifiability in their protocol, it's not complete. Based on the former, Fitzsimons et al.\cite{RN8} proposed a relatively complete verifiable version. In this protocol, the verifier encodes the computation task (including the verification mechanism) into a series of single qubits, and then executes BQC. According to the results, it can be verified whether the computation has been correctly executed. In addition, several BQC protocols\cite{RN2, RN8, RN9} verifies the correctness of the input of BQC by checking the trap qubits randomly hidden in the input state. For MOBQC\cite{RN7}, it's proposed to verify graph states\cite{RN10,RN11,RN12,RN13,RN14}. Stabilizer testing\cite{RN14} is a verification technology of the graph state without setting traps. The server generates graph states and send them to the client, and the latter then directly measure stabilizers on the sent graph states to verifies the correctness. However, these verifiable MOBQC protocols \cite{RN10,RN12,RN13,RN14} using stabilizer test are of high resource consumption, which is an obstacle to the development of scalable quantum computation. In 2021, Xu et al.\cite{RN15} proposed a verifiable BQC protocol based on entanglement witnesses, which effectively reduces the resource consumption of verification by measuring entanglement witnesses\cite{RN16} that can detect the graph states.

The above verifiable protocols only allow Alice to verify Bob's honesty, which is called \textit{private verifiability}. However, private verifiability has the following problems: Alice can detect any dishonest behavior of Bob, but the detection results can not make any third party trusted; even if Bob is honest, he can be framed by Alice. In 2016, on the basis of unconditionally verifiable BQC protocol\cite{RN8}, Kentaro\cite{RN17} proposed the concept of \textit{public verifiability} and provided a corresponding protocol based on classical cryptography. In 2019, Sato et al.\cite{RN18} chose to insert a trusted third party as the arbiter to build an arbitrable BQC protocol which realizes public verifiability in a sense. However, the public verifiability depends on the arbiter, which is determined in advance and participates in the whole process, thus the protocol isn't true public verifiable in a sense.

In this paper, inspired by the verifiable mechanism based on entanglement witnesses, we propose a public verifiable MOBQC protocol. The third-party verifier is randomly selected from other clients rather than a specific arbiter, so as to achieve public verifiability in the true sense. In addition, $2$-colorable graphs and entanglement witnesses are introduced to reduce resource consumption. Compared with the number of local measurements ($O(n^{2n+5})$) and of copies of the resource states ($O(n^{2n+5}2^n)$) of Sato et al.\cite{RN18}, our protocol have obvious advantages ($O(n^3\log{n})$ and $O(n^2\log{n})$ respectively). We also consider the communication error and give some error mitigation schemes.

The rest of this paper is organized as follows: In Sect.~\ref{sec2}, we briefly introduce $2$-colorable graph states, entanglement witnesses, and MOBQC. The protocol is presented in Sect.~\ref{sec3} and analyzed in Sect.~\ref{sec4}. Error propagation and mitigation is analyzed in Sect.~\ref{sec5}. The paper concludes with Sect.~\ref{sec6}.

\section{Preliminaries}\label{sec2}

In this section, we briefly introduce $2$-colorable graph states and the entanglement witnesses of them, and then review the basic steps of measurement-only BQC.

\subsection{2-colorable graph state}\label{subsec2.1}

Given an undirected simple graph $G$ with $n$ vertices $i\in V$ and several edges $(i,j)\in E=V\times V$, if all vertices of it can be divided into at least $m$ disjoint subsets ${S_1}, {S_2}, \cdots, {S_m}$, where there is no edge between any pair of vertices in any $S_j,j=1,2,\cdots,m$, then we call $G$ an $m$-colorable graph. We use $n$ qubits to represent vertices of $G$, and the graph state $\left\lvert G \right\rangle $ corresponding to $G$ is defined as $\left\lvert G \right\rangle  = \left( {\prod\nolimits_{\left( {i,j} \right) \in E} {{U_{ij}}} } \right){\left\lvert  +  \right\rangle ^{ \otimes n}}$, where $\left\lvert  +  \right\rangle  = {1 \over {\sqrt 2 }}\left( {\left\lvert 0 \right\rangle  + \left\lvert 1 \right\rangle } \right)$ is the initial state of each vertex and ${U_{ij}}$ is the controlled-$Z$ gate $\left\lvert 0 \right\rangle \left\langle 0 \right\lvert \otimes I + \left\lvert 1 \right\rangle \left\langle 1 \right\lvert \otimes Z$ performed on vertices $i$ and $j$, where $I$ is identity operator and $Z$ is Pauli operator $\sigma_z$. On the other hand, there are $n$ stabilizer ${g_i} = {X_i}\prod\nolimits_{k \in N\left( i \right)} {{Z_k}} $ of $\left\lvert G \right\rangle $, i.e., $g_i\left\lvert G \right\rangle=\left\lvert G \right\rangle$, where $i = 1,2, \cdots,n$, $N\left( i \right)$ is the adjacency points set of vertex $i$ and $X_i, Z_j$ are the Pauli operators $\sigma_x,\sigma_z$ performed on $i,j$ respectively. 

In this paper, we only consider $2$-colorable graph states which are widely used as resource states of BQC, such as brickwork state\cite{RN6} and Raussendorf-Harrington-Goyal (RHG) state\cite{RN19}, of which the preparation and verification are of research value. An example of a $2$-colorable graph is shown in Figure~\ref{fig1}.

\begin{figure}[h]
\centering
\includegraphics[width=0.5\textwidth]{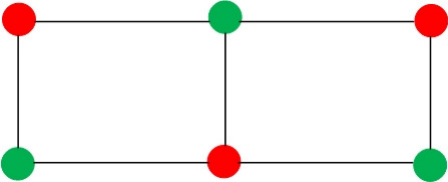}
\caption{A $2$-colorable graph as an example, where red vertices belong to $S_1$ and green vertices belong to $S_2$.}\label{fig1}
\end{figure}

\subsection{Entanglement witness}\label{subsec2.2}

An entanglement witness $W$ is an observable measurement which satisfies: (1) For all separable states ${\varrho _s} $, $tr\left ({W {\varrho_s}} \right) \ge 0 $; (2) At least one entangled state ${\varrho _e} $ satisfies $tr \left ({W{\varrho _e}} \right) < 0 $, where $tr\left (\cdot\right)$ represents matrix trace; then we say that ${\varrho _e} $ is detected by $W$. For an $n$-qubit graph state and some states close to it, based on the colorability of the graph, some witnesses with constant measurement times are proposed. The following $W^{\left(2 \right)} $ is a witness of a $2$-colorable graph state $\left\lvert G \right\rangle$\cite{RN16}:
\begin{equation}\label{eq1}
{W^{\left( 2 \right)}} = 3I - 2\left[ {\prod\limits_{i \in {S_1}} {{{{g_i} + I} \over 2}}  + \prod\limits_{i \in S2} {{{{g_i} + I} \over 2}} } \right],
\end{equation}
where $S_1,S_2$ are two divided sets of the graph. According to the structure of the witness, for a given $2$-colorable graph state, only two measuring settings are needed, and the $j$-th setting is observable $\prod\limits_{i \in {S_j}} {{g_i}} $. The two settings corresponding to the two-colorable graph in Figure~\ref{fig1} are shown in Figure~\ref{fig2}. For the $j$-th measuring setting $\prod\limits_{i \in {S_j}} {{g_i}} $, we only need to measure the qubits corresponding to ${S_j} $, and then measure the qubits corresponding to another subset $\overline{S_j} $ according to the adjacency relationship with $S_j$. Therefore, a setting $\prod\limits_{i \in {S_j}} {{g_i}} $ only needs $O(n)$ local measurement times.

\begin{figure}[H]
	\centering 
	\vspace{-0.35cm} 
	\subfigtopskip=2pt 
	\subfigbottomskip=2pt 
	\subfigcapskip=-5pt 
	\subfigure[]{
		\includegraphics[width=0.45\linewidth]{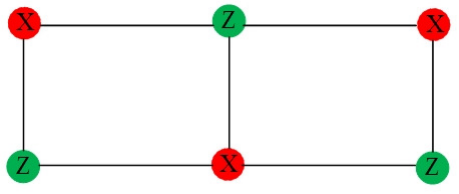}}
	\quad 
	\subfigure[]{
		\includegraphics[width=0.45\linewidth]{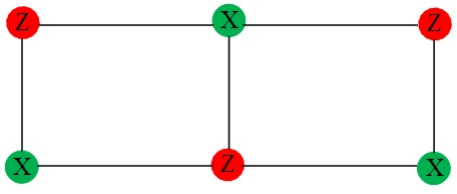}}
	\caption{The two settings corresponding to the graph in Figure~\ref{fig1}, where all red vertices belong to $S_1$ and all green vertices belong to $S_2$. (a)(b) are the observables $\prod\limits_{i \in {S_1}} {{g_i}} ,\prod\limits_{i \in {S_2}} {{g_i}} $ corresponding to $S_1,S_2$ respectively.}
	\label{fig2}
\end{figure}

\subsection{Measurement-only blind quantum computation}\label{subsec2.3}

In MOBQC, the server Bob only needs to prepare the general resource state, and the client Alice only needs to perform quantum measurement. The protocol steps are as follows: Bob prepares the general resource state and then sends the prepared state particles to Alice via quantum channel, then Alice measures the sent particles on the basis determined by her algorithm. The verification of this model is generally aimed at the correctness of the resource state. Bob is often required to prepare multiple copies of the resource state, and some of which are used for verification and one of the rest is used for calculation, as shown in Figure~\ref{fig3}.

\begin{figure}[ht]
\centering
\includegraphics[width=0.8\textwidth]{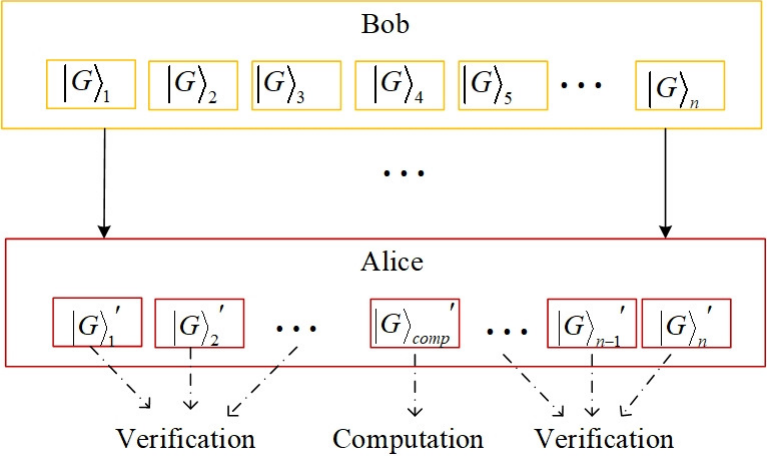}
\caption{The verification process of MOBQC.}\label{fig3}
\end{figure}

\section{Public Verifiable Measurement-only Blind Quantum Computation based on entanglement witnesses}\label{sec3}

\subsection{Verification algorithm}\label{subsec3.1}

Inspired from Xu et al.'s verification mechanism\cite{RN15}, we present a verification algorithm to verify the correctness of the prepared graph state. Given a target graph state $\left \lvert G\right \rangle$ corresponding to a $2$-colorable graph $G$, and an unknown state $\varrho$ to be verified. The two divided subsets of $G$ are denoted as $S_1,S_2$, and the verification process is shown in Algorithm~\ref{algo1}. 

\begin{algorithm}[ht]
\caption{\textbf{Verification Algorithm}: verify the correctness of the prepared graph state $\varrho$}\label{algo1}
\begin{algorithmic}[1]
\Require $2K$ $n$-qubit registers, a required graph state $\left \lvert G\right \rangle$, and a constant $0 \le C \le 2K $.
\Ensure The $2K$ registers store the same prepared $n$-qubit state $\varrho$.
\State Select $K$ registers independently, evenly and randomly from $2K$ registers and mark it as the $1$-st group, then mark the rest $K$ registers as the $2$-nd group.
\State Measure observable $\prod\limits_{i \in {S_j}} {{g_j}} $ on each register in the $j$-th group.
\State Calculate $M_j^\varrho $ by  $M_j^\varrho  = \prod\limits_{i \in {S_j}} {{{{x_i}\prod\nolimits_{k \in N\left( i \right)} {{z_k}}  + 1} \over 2}} $, where $x_i,z_j$ are the results of measuring observables $X_i,Z_j$, and count the number of registers satisfy $M_j^\varrho  = 0$ as ${K_j}$.
\If{${K_1} + {K_2} \le C$}\label{algln2}
        \State Accept.
\Else
        \State Reject.
\EndIf
\end{algorithmic}
\end{algorithm}

In Algorithm~\ref{algo1}, the condition constant $C$ is determined to make sure the fidelity between the prepared state $\varrho$ and the required state $\left \lvert  G\right \rangle$ is high enough. Considering the fidelity estimation process, $C$ isn't fixed, but varies with the order of the verifier, i.e., $C$ will be different for the third-party verifier from the client in our protocol. Therefore, we set $C$ as a pending parameter so as to ensures the scalability of the verification. Based on the above, Algorithm~\ref{algo1} can be applied to public verification. 

\subsection{Proposed protocol}\label{subsec3.2}

In Sato et al.'s protocol\cite{RN18}, Charlie, the third-party arbiter, can arbitrate in case of a dispute between the server Bob and the client Alice. However, the verifier (Charlie) is determined in advance and participates in the whole process, which isn't a true third party independent with Bob and Alice. To achieve a true public verification, i.e., any third party can participate in verification, we removed Charlie's role of third-party verifier, but only retained its storage capacity of quantum state, therefore it's renamed the storage center. The third party that participates in the public verification will be selected from other clients Alice$_2$, Alice$_3$,$\cdots$, Alice$_l$ randomly, where $l$ is the total number of clients. As shown in Figure~\ref{fig4}, there are three parties in the protocol: the server Bob is responsible for preparing the graph states; the set A = \{Alice$_1$, Alice$_2$,$\cdots$, Alice$_l$\} is a set of clients of quantum computation, which are all legal users\cite{RN22, RN23} registered with Charlie; the storage center Charlie has the ability to store quantum states and is responsible for distributing the graph states prepared by Bob to the client and selecting the third-party verifier, and is ensured honesty. When the  protocol is executed between Alice$_1$ and Bob, any other client Alice$_t$ in A where $ t \in \{ {\rm{2}},{\rm{ }} \ldots {\rm{ }},{\rm{ }}l\} $ can perform public verification.

\begin{figure}[ht]
\centering
\includegraphics[width=0.8\textwidth]{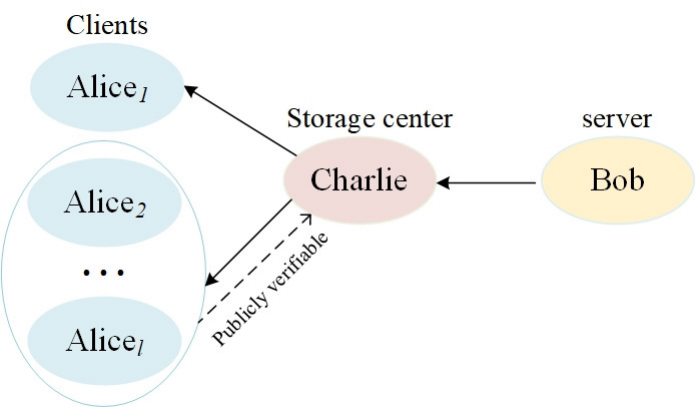}
\caption{The tripartite relationship in the proposed protocol.}\label{fig4}
\end{figure}

The graph states used in our protocol all correspond to $2$-colorable graphs. Taking client Alice$_1$, who has a computing request, as an example, the specific steps are as follows (also shown in Figure~\ref{fig5}):

\begin{enumerate}[Step 1]
\item Alice$_1$ sends preparation request to Charlie, then Charlie forwards it to Bob, where the graph state $\left \lvert G\right 
\rangle$ requested is an $n$-qubit state corresponding to a $2$-colorable graph and $n \ge 6$.
\item Bob prepares a $5Kn$-qubit state ${\left\lvert G \right\rangle ^{5K}}$, where$K = \left\lceil {{n^2}\log n} \right\rceil ,\left \lceil  \cdot  \right\rceil $  is ceiling function. Then he send it to Charlie one by one qubit.
\item Charlie divides the state sent from Bob into $5K$ $n$-qubit states $\left \lvert G \right \rangle$ in turn and stores them in $n$-qubit registers respectively. He selects $2K$ registers independently, evenly and randomly from these $5K$ registers and keeps them, and then sends the rest $3K$ to Alice$_1$ in turn.
\item Alice$_1$ divides the states sent from Charlie into $3K$ $n$-qubit registers and selects $2K$ registers independently, evenly and randomly from them, then executes Verification Algorithm (see Algorithm~\ref{algo1}) where $C = {\frac{K}{2n}}$.
\item If it accepts, Alice$_1$ considers Bob honest and proceeds to the next step, and otherwise considers Bob dishonest and refuses to pay for services.
\item Alice$_1$ randomly selects one register from the remaining $K$ registers and discards the others, then uses this register to perform MBQC, i.e., measures particles on the basis determined by her algorithm.
\item If Alice$_1$ claims that Bob is dishonest, Bob can ask Charlie for public verification, and then Charlie randomly selects a third party Alice$_t$ from Alice$_2$, Alice$_3$,$\cdots$, Alice$_l$ to send verification request.
\item If Alice$_t$ accepts the request, Charlie sends the $2K$ copies in his hand to Alice$_t$ in turn. According to the graph state type, Alice$_t$ executes Verification Algorithm, where $C = \frac{3K}{4n}$. If it accepts, Alice$_t$ claims that Alice$_1$ is dishonest, otherwise Bob is dishonest.
\end{enumerate}

\begin{figure}[ht]
\centering
\includegraphics[width=0.8\textwidth]{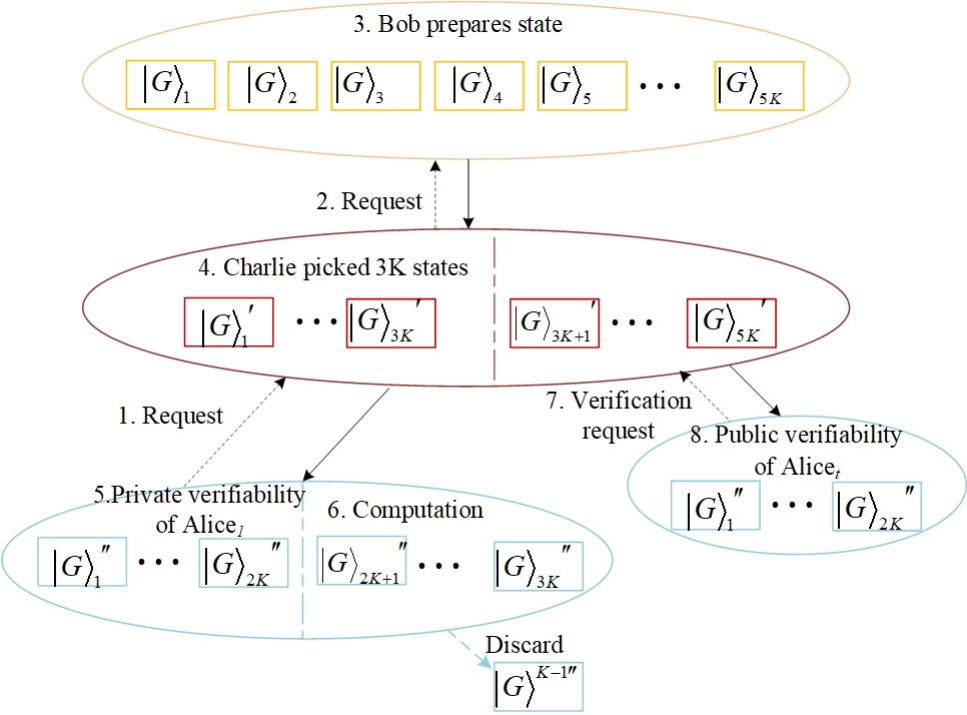}
\caption{The process of public verifiable protocol. Solid lines represent graph states transmission and dotted lines represent requests transmission. }\label{fig5}
\end{figure}

\section{Performance analysis}\label{sec4}

\subsection{Completeness analysis}\label{subsec4.1}

Completeness means that when Bob faithfully prepares the required graph state, it must be accepted by Alice$_1$ or Alice$_t$ with high probability. In Step~3 of Algorithm~\ref{algo1}, the random variable $M_j^\varrho  = \prod\limits_{i \in {S_j}} {{{{x_i}\prod\nolimits_{k \in N\left( i \right)} {{z_k}}  + 1} \over 2}}  \in \left\{ {0,1} \right\}$ to be calculated is the measurement result of $\prod\limits_{i \in {S_j}} {{{{g_i} + I} \over 2}} $. According to quantum measurement theory, we have
\begin{equation}\label{eq2}
\overline M _j^\varrho  = Tr\left( {\prod\limits_{i \in {S_j}} {{{{g_i} + I} \over 2}} \varrho } \right),
\end{equation}
where $\overline M _j^\varrho$ is the mathematical expectation of $M_j^\varrho$. Assume Bob prepares the correct state $\varrho  = \left\lvert G \right\rangle \left\langle G \right\lvert$, then we have 
\begin{equation}
\overline M _j^\varrho  = Tr\left( {\prod\limits_{i \in S_j} {\frac{g_i + I} {2}} \left\lvert G \right\rangle \left\langle G \right\lvert} \right) =Tr\left(\left \lvert G\right\rangle \left\langle G \right  \lvert  \right) = 1.
\end{equation}
We have $M _j^\varrho\in \{0,1\}$, which means for each register $k \in {\Pi ^{\left( j \right)}}$, $M_j^{{\varrho _k}} = 1$, i.e.,${K_j} = 0$. Therefore $\forall 0 \le C \le 2K, {K_1} + {K_2} \le C$, i.e., it must be accepted, leading to the completeness.

\subsection{Soundness analysis}\label{subsec4.2}

Soundness means that if Alice$_1$ or Alice$_t$ accepts the state $\varrho$ prepared by Bob, it must close to the graph state required by Alice$_1$ with high probability. Fidelity $F = \left \langle G \right \lvert \varrho \left \lvert G \right \rangle $ is generally used to measure the closeness. Fidelity estimation is based on the following inequality \cite{RN15}:
\begin{equation}\label{eq4}
F\ge \frac{1}{2} -\frac{1}{2}Tr\left( {W^{\left( 2 \right)}\varrho } \right).
\end{equation} In our protocol, the expectation $Tr\left ({W^{\left (2 \right)} \varrho} \right) $is obtained by measuring the entanglement witness $W^{\left (2 \right)}$ detecting $\left \lvert G\right \rangle$ to determine whether the state $\varrho$ is close to $\left \lvert G\right \rangle$, so as to verifies the behavior of Bob. We have the following theorem about the soundness of our protocol:

\begin{theorem}\label{theorem}  
In our protocol,
\begin{enumerate}[(1)]
\item When Alice$_t$ hasn't involved in arbitration, if Alice$_1$ measures ${K_1} + {K_2} \le \frac{K}{2n}$, then we have 
\begin{equation}\label{eq5}
 F\ge 1 - {{4\sqrt {\lambda _1}  + 1} \over n}
 \end{equation}
 with a probability 
\begin{equation}\label{eq6}
P\ge 1 - 4n^{-\frac{\lambda _1 }{2}},
\end{equation} 
where $\lambda_1$ is arbitrary constant satisfying $\log_n{16} \le \lambda_1\le \frac{{\left( n - 1\right)}^2}{16}$.
\item When Alice$_t$ involves in arbitration, if Alice$_t$ measures ${K_1} + {K_2} \le \frac{3K}{4n}$, then we have 
\begin{equation}\label{eq7}
 F\ge 1 - {{3\sqrt {\lambda_2}  + 1} \over n}
 \end{equation}
 with a probability 
\begin{equation}\label{eq8}
P\ge 1 - 4n^{-\lambda _2 },
\end{equation} 
where $\lambda_2$ is arbitrary constant satisfying $\log_n{4} \le \lambda_2\le \frac{{\left( n - 1\right)}^2}{9}$.
\end{enumerate}
\end{theorem}

\begin{proof}
In the theorem, (1) has been proved \cite{RN15} and we only need to prove (2), using some existing probability inequalities. If we perform the $j$-th measurement setting on the rest $3K$ registers, then for each group of $K$ registers selected by Charlie, we can obtain an upper bound of the number of registers satisfying $M_j^{\varrho} = 0$ in the rest $3K$ registers measured and the relevant confidence probability, Hence, the lower bound of $\sum\limits_{k = 1}^{3K} {M_1^{\varrho _k}}$ is directly given. We then obtain a lower bound of $\overline{M}_1^{\varrho}$ and the relevant confidence probability. By Eq.~\ref{eq2} we can obtain a lower bound of $ Tr\left( {\prod\limits_{i \in {S_j}} {{{{g_i} + I} \over 2}} \varrho } \right)$. By Eq.~\ref{eq1} and Eq.~\ref{eq4}, we finally prove that the fidelity $F$ satisfies a lower bound with a certain confidence probability $P$. See Appendix~\ref{secA} for details.
\end{proof}

Therefore, if the verification is passed, then the state prepared by Bob is close to the required graph state with high probability, leading to the soundness.

In the protocol, $K=O\left ({n ^ 2} \log n \right)$ so that the probability in the theorem is $1 - O \left ({n ^ {-\lambda}}\right) $ for a constant $\lambda$, which is high enough; we set $n\ge 6$ so that the above $\lambda$ exists. See Appendix~\ref{secA} for details.

\subsection{Efficiency analysis}\label{subsec4.3}

Efficiency refers to the number of copies of the resource state prepared by the protocol and the number of local measurements required. As mentioned above, we set $K=O\left ({n ^ 2} \log n \right)$ so that the probability is high enough. The detail parameters of our and Sato et al.'s protocol\cite{RN18} are shown in Table~\ref{table1}. 

\begin{table}[ht]
\begin{center}
\begin{minipage}{\textwidth}
\caption{The detail parameters of public verifiable MBQC protocols}\label{table1}%
\begin{tabular}{@{}llll@{}}
\toprule
Protocol & Copies numbers  & Parameters ranges & Observables measured\\
\midrule
Sato et al.'s   & $\left\lvert G \right\rangle ^{2k + m + 1}$   & $k \ge 4{n^2} - 1$, $m \ge \left( {2\ln 2} \right){K^n}{n^5}$  & $\prod\limits_{i = 1}^n {g_i^{{k_i}}} $  \\
Our  & ${\left\lvert G \right\rangle ^{5K}}$   & $K \ge \left\lceil {{n^2}\log n} \right\rceil $ & $\prod\limits_{i \in {S_j}}^n {{g_i}} $  \\
\botrule
\end{tabular}
\end{minipage}
\end{center}
\end{table}

We first consider the number of copies. For the same $n$-qubit graph state $\left\lvert G \right\rangle $, in Sato et al.'s protocol, the number of copies  is 
\begin{equation}
\begin{aligned}
2k + m + 1& \ge 8{n^2} - 2 + \left( {2\ln 2} \right){k^n}n + 1  \\
&= 8{n^2} - 2 + \left( {2\ln 2} \right){k^n}n + 1 \\
&= \Theta \left( {{n^{2n + 5}}} \right)
\end{aligned},
\end{equation}
i.e., $O\left( {{n^{2n + 5}}} \right)$ at least; in our protocol it is $5K \ge 5\left\lceil {{n^2}\log n} \right\rceil  = \Theta \left( {{n^2}\log n} \right)$, i.e., $O\left( {{n^2}\log n} \right)$ at least. Thus our protocol has advantages.

Then, the number of required local measurement is taken into account. In Sato et al.'s protocol, they measured stabilizers $\prod\limits_{i = 1}^n {g_i^{{k_i}}} $ on $k$ or $2k$ $n$-qubit graph states, where  $k = \left( {{k_1}{k_2} \ldots {k_n}} \right)$ is a string randomly selected. Considering that the local measurement decomposition of $\prod\limits_{i = 1}^n {g_i^{{k_i}}} $may be complex, and for a specific graph structure it may be $O\left( {{2^n}} \right)$\cite{RN24}, thus the number of local measurements is $O\left( {kn} \right) = O\left( {{n^{2n + 5}}{2^n}} \right)$; in our protocol, we measure observables $\prod\limits_{i \in {S_j}}^n {{g_i}} $ on $2K$ $n$-qubit graph states and for each graph state the number of the local measurements is $O(n)$ as mentioned above, thus the total number is $O\left( {Kn} \right) = O\left( {{n^3}\log n} \right)$. Therefore, our protocol still has advantages.

\section{Error propagation and mitigation}\label{sec5}

It is worth noting that all the above analyses are based on the fact that the quantum channel does not contain noise. However, the actual channel has a certain amount of noise, so there must be certain errors in the graph state's propagation. Since Bob is not assumed to be honest in our protocol (i.e., the sent graph state is not guaranteed to be correct), it is impossible to determine whether the graph state is disturbed by noise by comparing it with the target state. To mitigate the impact of noise, we have the following two methods.

\noindent (1) \textit{Use channel noise detection} For each $n$-qubit graph state $\left \lvert G \right \rangle$, the sender (Bob or Charlie) insert some additional qubits that are not entangled with the graph state into the $n$ qubits. These qubits' initial states are agreed in advance by both the sender and the receiver (Charlie or Alice), and they will be sent together with the graph state as a whole. In this way, if some noise is encountered during one transmission, the state of these extra qubits will be changed with a certain probability, and thus the noise can be detected. If the noise is considered to be too high, the receiver may reject the communication and request retransmission.

For example, assuming that $k$ additional qubits are in the initial state $\left \lvert 0\right \rangle$, bit-flip noise occurs in the quantum channel with probability $p$, and $r$ qubits are flipped to $\left \lvert 1\right \rangle$ after transmission. Since $r$ has a binomial distribution, its expectation is $kp$. By Azuma-Hoeffding bound\cite{RN21} (see Appendix~\ref{secA} for details), $\forall t>0$, we have 
\begin{equation}
    \Pr \left(  \left(1-\frac{r}{k}\right) - \left(1-p\right) \le t\right)=\Pr \left(  p\le\frac{r}{k}+t\right)  \ge 1 - \exp\left(-2kt^2\right).
\end{equation}
If $p_{th}$ is the noise threshold, then 
\begin{equation}
    \Pr \left(  p\le p_{th}\right)  \ge 1 - \exp\left(-2k\left(p_{th}-\frac{r}{k}\right)^2\right).
  \end{equation}
Let $1 - \exp\left(-2k\left(p_{th}-\frac{r}{k}\right)^2\right)\ge 99\%$, then $r\le kp_{th}-p\sqrt{k\ln{10}}=r_{th}$. For instance, let $k=5$, $p=1\%$, then $r_{th}=3.44$. If $r\le 3$, then we can say $\Pr \left(  p\le p_{th}\right)  \ge 99\%$. The larger the $k$, the smaller the $\frac{r_{th}-\left \lfloor r_{th}\right \rfloor}{r_{th}}\le \frac{1}{r_{th}}$, and the tighter the upper bound. Other noise types can be detected similarly, and only the corresponding initial states and measurement bases need to be agreed.

\noindent (2) \textit{Use fault-tolerant quantum computing (FTQC)} As mentioned by Morimae et al.\cite{RN7}, using a computational model that can handle particle losses can effectively mitigate noise. An $[\![n,k,d]\!]$ quantum error-correcting codes (QECC) encodes $n$ physical qubits into $k$ logical qubits, and a QECC with distance $d$ can correct up to $\frac{d-1}{2}$ errors on arbitrary qubits\cite{RN25}. The entanglement of the graph state will not be destroyed in a qubit stabilizer QECC scheme, because it is not necessary to really know the initial state of the target qubit, but only to measure and compare the relative change between the physical qubits. On the other hand, only the measurements and quantum gates of single-qubit Pauli operators $X,Y,Z$ are required, which means that even receiver with weak quantum ability can implement it. In existing fault-tolerant quantum computing, the noise threshold can even reach $24.9\%$\cite{RN26}.

Note that the above two methods can be used in combination because they are independent of each other. First, the channel noise detection can ensure that the noise factor is lower than a certain threshold, and then the fault tolerance mechanism can correct small errors. By using the two methods, the error caused by channel noise can be mitigated to a certain extent. Of course, when the noise reaches a certain level, even retransmission will fail.

\section{Conclusion}\label{sec6}

In this paper, we proposes a public verifiable measurement-only blind quantum computation protocol. By introducing a storage center, it allows the third-party verifier to be any other client randomly selected. Compared with the previous protocol, our protocol is public verifiable in the true sense. In the protocol, the fidelity estimation between arbitrary states and graph states are realized by measuring the entanglement witnesses detecting the graph states. Without loss of completeness and soundness, the nature of $2$-colorable graph states reduce the number of local measurements ($O\left(n^3 \log n \right)$) and the number of copies of the graph states resources ($O\left (n^2 \log n \right) $). Compared with the arbitrable protocol of Sato et al.\cite{RN18} (the number of local measurements is $O(n^{2n+5})$ and of copies of the resource states is $O(n^{2n+5}2^n)$), our protocol has obvious advantages. We also consider the communication error and give some error mitigation schemes.

On the other hand, since we have only considered $2$-colorable graphs, the proposed protocol is not applicable to arbitrary graph states. For more general graph states, more research is needed to further improve the efficiency and performance of existing schemes.

\bmhead{Acknowledgments} %\textbf{Acknowledgments.}
The authors would like to thank the anonymous reviewers and editors for their comments that improved the quality of this paper. This work is supported by the National Natural Science Foundation of China (62071240), the Innovation Program for Quantum Science and Technology (2021ZD0302902), and the Priority Academic Program Development of Jiangsu Higher Education Institutions (PAPD).

\begin{appendices}

\section{Proof of Theorem~\ref{theorem}}\label{secA}

In the theorem, (1) has been proved \cite{RN15} with a condition $n\ge 6$. Now we prove (2). At first we introduce the following two probability bounds which will be used in the analysis, where $\Pr \left (\cdot \right) $ represents the event probability and $\rm{E}\left (\cdot \right) $ represents the mathematical expectation:

\begin{enumerate}[(1)]
\item \textit{Serfling's bound\cite{RN20}} Given a set $Y = (Y_1,Y_2,...,Y_T )$ of $T$ binary random variables with $Y_k\in \{0, 1\}$ and two arbitrary positive integers $N$ and $K$ that satisfy $T = N + K$, select $K$ samples that are distinguished from each other independently, evenly and randomly from $Y$, and let $\Pi$ be the set of these samples, $\overline{\Pi}=Y-\Pi$, then $\forall 0<v<1$, we have 
\begin{equation}\label{eqA1}
			\begin{aligned}
&\Pr \left( {\sum\limits_{k \in \overline \Pi  } {{Y_k}}  \le {N \over K}\sum\limits_{k \in \Pi } {{Y_k}}  + Nv} \right) \\
		&\ge 1 - \exp \left( { - {{2{v^2}N{K^2}} \over {\left( {N + K} \right)\left( {K + 1} \right)}}} \right).
\end{aligned}\end{equation}

\item \textit{Azuma-Hoeffding bound\cite{RN21}} Given independent random variables ${\xi _1},{\xi _2}, \cdots ,{\xi _n}$ where ${\xi _i} \in \left[ {{a_i},{b_i}} \right], i=1,2,\cdots,n$, then $\forall t > 0$, we have 
\begin{equation}\label{eqA2}
		\begin{aligned}
&\Pr \left( {{{{\xi _1} + {\xi _2} +  \cdots  + {\xi _n}} \over n} - {\rm{E} }\left( {{{{\xi _1} + {\xi _2} +  \cdots  + {\xi _n}} \over n}} \right) \le t} \right)  \\
&\ge 1 - \exp \left( { - {{2{n^2}{t^2}} \over {\sum\limits_{i = 1}^n {\left( {{b_i} - {a_i}} \right)} }}} \right).
\end{aligned}\end{equation}
\end{enumerate}

For the first $K$ registers selected, we denote them as $\Pi^{(1)}$ and the rest $4K$ as $\overline{\Pi}^{(1)}$. Let $T = 5K, N=4K, Y_k = \left\{\begin{matrix}0,M_1^{\varrho'_k} = 1 \\1,M_1^{\varrho'_k} = 0  \end{matrix}\right.$, where $\varrho'_k$ is the state in the $k$-th register in $\Pi^{(1)}$ or $\overline{\Pi}^{(1)}$, then we have 
\begin{equation}\label{eqA3}
\Pr \left( {\sum\limits_{k \in {{\overline \Pi  }^{\left( 1 \right)}}} {{Y_k}}  \le {{4K} \over K}\sum\limits_{k \in {\Pi ^{\left( 1 \right)}}} {{Y_k}}  + 4Kv} \right) \ge 1 - \exp \left( { - {{2{v^2}4K{K^2}} \over {\left( {4K + K} \right)\left( {K + 1} \right)}}} \right)
\end{equation}
by Eq~\ref{eqA1}, which means if we perform the $j$-th measurement on the rest $4K$ registers, then the upper bound  of the number of the registers satisfying $M_1^\varrho  = 0$ (i.e., ${Y_k} = 1$) in ${\overline \Pi  ^{\left( 1 \right)}}$ is $4\sum\limits_{k \in {\Pi ^{\left( 1 \right)}}} {{Y_k}}  + 4Kv$, with the  probability on the right-side of Eq~\ref{eqA3}. 
Similarly, for the second $K$ registers selected, we denote them as $\Pi^{(2)}$ and the rest $3K$ as $\overline{\Pi}^{(2)}$. Let $T = 4K, N=3K, Y_k = \left\{
\begin{matrix}
{0,M_2^{\varrho'_k} = 1}\\
{1,M_2^{\varrho'_k} = 0}
\end{matrix} \right.$, where $\varrho'_k$ is the state in the $k$-th register in $\Pi^{(2)}$ or $\overline{\Pi}^{(2)}$, then we have 
\begin{equation}\label{eqA4}
\Pr \left( {\sum\limits_{k \in {{\overline \Pi  }^{2}}} {{Y_k}}  \le {{3K} \over K}\sum\limits_{k \in {\Pi ^{2}}} {{Y_k}}  + 3Kv} \right) \ge 1 - \exp \left( { - {{2{v^2}3K{K^2}} \over {\left( {3K + K} \right)\left( {K + 1} \right)}}} \right),
\end{equation}
which means if we perform the $j$-th measurement on the rest $3K$ registers, then the upper bound of the number of the registers satisfying $M_2^\varrho  = 0$ in ${\overline \Pi  ^{\left( 2 \right)}}$ is $4\sum\limits_{k \in {\Pi ^{\left( 2 \right)}}} {{Y_k}}  + 4Kv$, with the  probability on the right-side of Eq~\ref{eqA4}. In the protocol, any two clients do not trust each other, thus it can be considered that the rest $3K$ registers haven't been measured. If we perform the first measurement on the rest $3K$ registers, there will be $3K - \left( {4\sum\limits_{k \in {\Pi ^{\left( 1 \right)}}} {{Y_k}}  + 4Kv} \right)$ registers satisfying $M_1^\varrho  = 1$ at least, i.e., 
\begin{equation}\label{eqA5}
\sum\limits_{k = 1}^{3K} {M_1^{{\varrho _k}}}  \ge 3K - \left( {4\sum\limits_{k \in {\Pi ^{\left( 1 \right)}}} {{Y_k}}  + 4Kv} \right).
\end{equation}
Similarly, we have 
\begin{equation}\label{eqA6}
\sum\limits_{k = 1}^{3K} {M_2^{{\varrho _k}}}  \ge 3K - \left( {3\sum\limits_{k \in {\Pi ^{\left( 2 \right)}}} {{Y_k}}  + 3Kv} \right).
\end{equation}
Let $n = 3K, \xi _k = M_1^\varrho $ or $M_2^\varrho $, then by Eq~\ref{eqA2} we have 
\begin{equation}\label{eqA7}
\Pr \left( {{1 \over {3K}}\sum\limits_{k = 1}^{3K} {M_1^{{\varrho _k}}}  - \overline M _1^\varrho  \le t} \right) \ge 1 - \exp \left( { - 2 \cdot 3K{t^2}} \right).
\end{equation}
By $Tr\left( {\prod\limits_{i \in {S_1}} {{{{g_i} + I} \over 2}} \varrho } \right) = \overline M _1^\varrho $ and Eq~\ref{eqA5} we have 
\begin{equation}\label{eqA8}
\begin{aligned}
&\Pr \left( {Tr\left( {\prod\limits_{i \in {S_1}} {{{{g_i} + I} \over 2}} \varrho } \right) > 1 - {1 \over {3K}}\left( {4\sum\limits_{k \in {\Pi ^{\left( 1 \right)}}} {{Y_k}}  + 4Kv} \right) - t} \right) \\
&\ge 1 - \exp \left( { - 6K{t^2}} \right)
\end{aligned},
\end{equation} 
and similarly we have 
\begin{equation}\label{eqA9}
\begin{aligned}
&\Pr \left( {Tr\left( {\prod\limits_{i \in {S_2}} {{{{g_i} + I} \over 2}} \varrho } \right) > 1 - {1 \over {3K}}\left( {3\sum\limits_{k \in {\Pi ^{\left( 2 \right)}}} {{Y_k}}  + 3Kv} \right) - t} \right) \\
&\ge 1 - \exp \left( { - 6K{t^2}} \right)
\end{aligned}.
\end{equation}
Therefore, we have 
\begin{equation}\label{eqA10}
\begin{aligned}
  & F \ge {1 \over 2} - {1 \over 2}Tr\left( {{W^{\left( 2 \right)}}\varrho } \right)  \cr 
  &  = {1 \over 2} - {1 \over 2}Tr\left( {3I\varrho } \right) + {1 \over 2}Tr\left( {2\prod\limits_{i \in {S_1}} {{{{g_i} + I} \over 2}\varrho } } \right) + Tr\left( {2\prod\limits_{i \in {S_2}} {{{{g_i} + I} \over 2}\varrho } } \right)  \cr 
  &  =  - 1 + 1 - {1 \over {AK}}\left( {4\sum\limits_{k \in {\Pi ^{\left( 1 \right)}}} {{Y_k}}  + 4Kv} \right) - t + 1 - {1 \over {3K}}\left( {3\sum\limits_{k \in {\Pi ^{\left( 2 \right)}}} {{Y_k}}  + 3Kv} \right) - t  \cr 
  &  = 1 - \left( {2 + {1 \over 3}} \right)v - 2t - {1 \over {3K}}\left( {4\sum\limits_{k \in {\Pi ^{\left( 1 \right)}}} {{Y_k}}  + 3\sum\limits_{k \in {\Pi ^{\left( 2 \right)}}} {{Y_k}} } \right)  \cr 
  &  \ge 1 - \left( {2 + {1 \over 3}} \right)v - 2t - {4 \over {3K}}\left( {\sum\limits_{k \in {\Pi ^{\left( 1 \right)}}} {{Y_k}}  + \sum\limits_{k \in {\Pi ^{\left( 2 \right)}}} {{Y_k}} } \right)  \cr 
  &  = 1 - \left( {2 + {1 \over 3}} \right)v - 2t - {4 \over {3K}}\left( {{K_1} + {K_2}} \right) .
  \end{aligned}
\end{equation} 
with a probability 
\begin{equation}\label{eqA11}
\begin{aligned}
  & P \ge \left[ {1 - \exp \left( { - {{8{v^2}K} \over {5\left( {1 + {1 \over K}} \right)}}} \right)} \right]\left[ {1 - \exp \left( { - {{3{v^2}K} \over {2\left( {1 + {1 \over K}} \right)}}} \right)} \right]{\left[ {1 - \exp \left( { - 6K{t^2}} \right)} \right]^2}   \\
  &  \ge {\left[ {1 - \exp \left( { - K{v^2}} \right)} \right]^2}{\left[ {1 - \exp \left( { - 6K{t^2}} \right)} \right]^2}
  \end{aligned},
\end{equation}
where the second inequality in Eq~\ref{eqA11} holds as long as $K \ge 2$. Obviously $\sum\limits_{k \in {\Pi ^{\left( 1 \right)}}} {{Y_k}}  = {K_1}$ and $\sum\limits_{k \in {\Pi ^{\left( 2 \right)}}} {{Y_k}}  = {K_2}$ in Eq~\ref{eqA10}. To make $F = \left\langle G \right\lvert\rho \left\lvert G \right\rangle $ is $1 - O\left( {{1 \over n}} \right)$, which is high enough, we need that $v = O\left( {{1 \over n}} \right), t = O\left( {{1 \over n}} \right), {4 \over {3K}}\left( {{K_1} + {K_2}} \right) \le {1 \over n}$ which leads to $F = 1 - O\left( {{1 \over n}} \right)$. Therefore, we set $v = {{\sqrt {{\lambda _2}} } \over n}, t = {{\sqrt {{\lambda _2}} } \over {\sqrt 6 n}}$, then consider the acceptance condition ${K_1} + {K_2} \le {{3K} \over {4n}}$ in Algorithm~\ref{algo1}, we have 
\begin{equation}
F \ge 1 - {7 \over 3}{{\sqrt {{\lambda _2}} } \over n} - 2{{\sqrt {{\lambda _2}} } \over {\sqrt 6 n}} - {1 \over n} = 1 - {{\left( {{7 \over 3} + {2 \over {\sqrt 6 }}} \right)\sqrt {{\lambda _2}}  + 1} \over n} \ge 1 - {{3.15\sqrt {{\lambda _2}}  + 1} \over n}
\end{equation}
with a probability
\begin{equation}
P \ge {\left[ {1 - \exp \left( { - {{{\lambda _2}} \over {{n^2}}}K} \right)} \right]^4} \ge 1 - 4\exp \left( { - {{{\lambda _2}} \over {{n^2}}}K} \right) \ge 1 - 4\exp \left( { - {{{\lambda _2}} \over {{n^2}}}{n^2}\log n} \right).
\end{equation}
 To make the probability $P=1 - O \left ({n ^ {-\lambda}}\right) $ for a constant $\lambda$, we set $K = \left\lceil {{n^2}\log n} \right\rceil $, then 
 \begin{equation}
 P \ge 1 - 4\exp \left( { - {{{\lambda _2}} \over {{n^2}}}{n^2}\log n} \right) = 1 - 4{n^{ - {\lambda _2}}},
 \end{equation}which is high enough. The condition for above $F,P$ both to be positive is ${\log _n}4 \le {\lambda _2} \le {{{{\left( {n - 1} \right)}^2}} \over {10}}$, where $n \ge 5$. When $n \ge 5$, we have ${n^2} > {{{{\left( {n - 1} \right)}^2}} \over {10}}$, thus ${\lambda _2} < {n^2}$, then $v = {{\sqrt {{\lambda _2}} } \over n} < 1$. Consider the condition of (1), we have $n\ge 6$.

\end{appendices}

%\bibliography{bibliography}

\end{document}